\documentclass[sigconf]{aamas}

\usepackage{balance} 
\usepackage{algorithm}
\usepackage{algorithmic}
\usepackage{amsmath}
\usepackage{amsthm}
\usepackage{bm}
\usepackage{subcaption}
\usepackage{derivative}
\usepackage{color}
\usepackage{natbib}

\newtheorem{defn}{Definition}
\newtheorem{thm}{Theorem}
\newtheorem{lem}{Lemma}

\newtheorem{fact}{Fact}

\DeclareMathOperator*{\argmax}{arg\,max}

\setcopyright{ifaamas}
\acmConference[AAMAS '26]{Full verson of the extended abstract in Proc.\@ of the 25th International Conference
on Autonomous Agents and Multiagent Systems (AAMAS 2026)}{May 25--29, 2026}
{Paphos, Cyprus}{C.~Amato, L.~Dennis, V.~Mascardi, J.~Thangarajah (eds.)}
\copyrightyear{2026}
\acmYear{2026}
\acmDOI{}
\acmPrice{}
\acmISBN{}

\acmSubmissionID{68}

\title[PSNE in a Two-Party Policy Competition: Existence and Algorithmic Approaches]{Computing Pure-Strategy Nash Equilibria in a Two-Party Policy Competition: Existence and Algorithmic Approaches}

\author{Chuang-Chieh Lin}
\affiliation{
  \institution{Department of Computer Science and Engineering, National Taiwan Ocean University}
  \city{Keelung City}
  \country{Taiwan}}
\email{josephcclin@mail.ntou.edu.tw}

\author{Chi-Jen Lu}
\affiliation{
  \institution{Institute of Information Science, \\Academia Sinica}
  \city{Taipei City}
  \country{Taiwan}}
\email{cjlu@iis.sinica.edu.tw}

\author{Po-An Chen}
\affiliation{
  \institution{Institute of Information Management, \\National Yang Ming Chiao Tung University}
  \city{Hsinchu City}
  \country{Taiwan}}
\email{poanchen@nycu.edu.tw}

\author{Chih-Chieh Hung}
\affiliation{
  \institution{Department of Management Information Systems, \\National Chung Hsing University}
  \city{Taichung City}
  \country{Taiwan}}
\email{smalloshin@nchu.edu.tw}

\begin{abstract}
We formulate two-party policy competition as a two-player non-cooperative game, generalizing Lin et al.'s work (2021). 
Each party selects a real-valued policy vector as its strategy from a compact subset of Euclidean space, and a voter's utility for a policy is given by the inner product with their preference vector. To capture the uncertainty in the competition, we assume that a policy’s winning probability increases monotonically with its total utility across all voters, and we formalize this via an affine isotonic function.  
A player's payoff is defined as the expected utility received by its supporters. In this work, we first test and validate the isotonicity hypothesis through voting simulations. Next, we prove the existence of a pure-strategy Nash equilibrium (PSNE) in both one- and multi-dimensional settings. Although we construct a counterexample demonstrating the game’s non-monotonicity, our experiments show that a decentralized gradient-based algorithm typically converges rapidly to an approximate PSNE. Finally, we present a grid-based search algorithm that finds an $\varepsilon$-approximate PSNE of the game in time polynomial in the input size and~$1/\varepsilon$. 
\end{abstract}


\newcommand{\BibTeX}{\rm B\kern-.05em{\sc i\kern-.025em b}\kern-.08em\TeX}

\begin{document}


\pagestyle{fancy}
\fancyhead{}


\maketitle 


\section{Introduction and Motivation}
\label{sec:intro}

Since the 1920s, political competition has been investigated in the paradigm of {\em Spatial Theory of Voting}~\cite{Hot29,Dow57,LRL2007,Pal84,Web92}. 
In such settings, there are two parties and voters with single-peaked preferences over a unidimensional metric space. Each party 
aims to choose a vector that is as close as possible to voters' preferences. In the 1950s, Duverger's law suggested \emph{plurality voting} favors the two-party system~\cite{Duv54}. 
Dellis~\cite{Del2013} explained why a two-party system emerges under plurality voting. 
These works motivated the investigation of the efficiency of a two-party system in~\cite{LLC2021}, in which a game of two players is formulated, 
each player corresponds to a political party as a collective entity, and each player has a set of candidates to nominate as a (pure) strategy to play.
In~\cite{LLC2021}, society is composed of supporters of party~$A$ and party~$B$, respectively. Unlike the ``winner-takes-all'' assumption, here each candidate benefits the supporters of both its own party and the opposing party. In particular, the winning probability is proportional to the total utility that a candidate brings to all voters. Hence the payoff of a party is the expected utility received by its supporters. Much progress was made in~\cite{LLC2021}, such as the existence of a pure-strategy Nash equilibrium (PSNE), and the lower and upper bounds on the price of anarchy of the game.

In a democratic country, parties compete not only through candidate nomination, but also by policy proposals, 
which can be multi-dimensional by considering various issues. 
This motivates the work here to further investigate the stability of two-party competition via policy proposals. 
In this paper, instead of regarding candidates as the ``discrete'' strategies to play, 
we consider a more generalized setting where a candidate can establish or propose a policy in the two-party competition. 
Here we regard a policy as a real vector in a compact high-dimensional Euclidean subspace. 
For example, if there are $k\in \mathbb{N}$ issues to take into consideration, the policy is a vector in~$\mathbb{R}^k$. 
The utility brought to a voter from a policy is measured by the inner product of the policy and the voter's preference vector. 
Hence, a policy proposed by one party can benefit both its own supporters and the other party's supporters. 
This setting clearly generalizes prior work~\cite{LLC2021}, and indeed, 
it is nontrivial from the perspective of the size of the strategy set, which is infinitely large now. 
Though a mixed-strategy Nash equilibrium always exists in a finite game~\cite{nash_1950,nash_1951}, 
a PSNE is not guaranteed to always exist (e.g., the well-known two-player Matching Pennies game) and 
determining whether a PSNE exists in a finite game of~$n$ players is {\sf NP}-complete~\cite{AGS2011}.  

\subsection{Our Contributions}
\label{subsec:contribution}

Our results can be positioned as a complementary alternative to classic distance-based spatial models. 
The inner product of a policy with voters' preference vectors encodes directional alignment, which can be \emph{beneficial} or \emph{harmful}, 
and intensity, which reflects \emph{how strongly} a policy affects the voters. 
It differs from distance models that penalize magnitude and do not distinguish the directions. Let us consider the following example. Consider a voter whose preference vector is $\mathbf{q} = (0.1,\,-0.1)$, indicating that she favors increases in education spending (first coordinate) and reductions in taxes (second coordinate). Now compare two policies $\mathbf{z}_B = (0.3,\,-0.3), \mathbf{z}_H = (-0.1,\,0.1)$. Both policies are equally distant from the voter’s preference vector because 
$\|\mathbf{z}_B - \mathbf{q}\| = \|\mathbf{z}_H - \mathbf{q}\| = \sqrt{0.08}$. However, their directional alignment with~$\mathbf{q}$ differs so that the inner products $\mathbf{q}^\top \mathbf{z}_B = 0.06 > 0$ while $\mathbf{q}^\top \mathbf{z}_H = -0.02 < 0$. Thus, although $\mathbf{z}_B$ and $\mathbf{z}_H$ are equally distant from the voter’s preference vector, the inner products reveal that $\mathbf{z}_B$ is beneficial (aligned with preferences), 
whereas $\mathbf{z}_H$ is harmful (misaligned). 
This illustrates how the inner-product formulation captures directional alignment that distance-based models cannot distinguish. 
Secondly, a party faces a dilemma in \emph{satisfying} its supporters by proposing a policy that aligns them or \emph{compromising} to appeal to the whole set of voters in order to increase its chance of winning. To the best of our knowledge, such a point of view has not been addressed before. 
Moreover, classic spatial voting shows instability in the high dimensions under majority rule~\cite{Dug08}. 
By contrast, our objective and isotonicity of winning restore stability by guaranteeing the existence of a PSNE, which provides a predictive benchmark. Our findings provide a computable prediction of equilibrium platforms for real campaigns, given fully observed preference vectors.

Specifically, our contributions are summarized as follows. 
\begin{itemize}
    \item We provide experimental simulations to support the isotonicity hypothesis on the winning probability computation, 
    thereby validating the framework in~\cite{LLC2021} as well.  
    \item We propose a closed-form solution of a PSNE of the two-party policy competition game in the one-dimensional setting. 
    \item We prove that a PSNE exists in the two-party policy competition game even in the multi-dimensional setting. 
    \item We give a counterexample showing the game is not monotone in general. Since the game is not of zero-sum, and it is neither convex nor concave, the existing gradient ascent and its accelerated variants, such as optimistic gradient ascent and extragradient methods (e.g., see~\cite{farina2022near,cai2023doubly,jelassi2020extragradient,Gorbunov2022a,Gorbunov2022b}), do not necessarily converge to a PSNE.  
    However, our experiments show that a vanilla gradient-based algorithm typically converges reasonably quickly, and most of the convergences correspond to approximate PSNE in practice. 
    \item We present a polynomial time grid-based search algorithm which discretizes the policy domain to find an $\varepsilon$-approximate PSNE. 
\end{itemize}

\paragraph{Organization of this paper.} Related work is summarized in Sect.~\ref{subsec:related_work}. Our model and preliminaries are introduced in Sect.~\ref{sec:models}. The PSNE existence is proved in Sect.~\ref{sec:results_I} and~\ref{sec:PSNE_general_case} for one-dimensional and multi-dimensional settings, respectively. In Sect.~\ref{sec:nonmonotone_nonLipschitz}, monotonicity of the game is disproved, and experimental simulations of a gradient-based algorithm are discussed. Then, in Sect.~\ref{sec:discretization}, we present a grid-based search algorithm to find an approximate PSNE. Finally, future work is discussed in Sect.~\ref{sec:future}.

\subsection{Related work}
\label{subsec:related_work}

\subsubsection{On Candidate Nominations.}

Lin et al.~\cite{LLC2021} proposed the \emph{two-party election game}, in which each party is modeled as a strategic player whose strategies are its candidates and the payoff is the expected utility received by its supporters. They showed that under the isotonicity hypothesis---that a candidate wins with higher probability if it brings more utility to all voters, the game always admits a PSNE using either a linear or softmax winning-probability function and the price of anarchy w.r.t. PSNE is constantly bounded. In the following paragraphs, we discuss related work 
relevant to the game on two-party competition.

\paragraph{On spatial theory of voting.} Dating back to~\cite{Hot29}, political competition was modeled based on {\em Spatial Theory of Voting}~\cite{Hot29,Dow57,Pal84,Web92}, in which society consists of two parties and voters with single-peaked preferences over a one-dimensional metric space. Each party chooses a policy which is as close to voters' preferences as possible. The Spatial Theory of Voting implies that the parties' strategies can be determined by the median voter's preference, when policies are assumed in a one-dimensional space. However, a PSNE is not guaranteed to exist for policies over a multi-dimensional space~\cite{Dug08}.

\paragraph{On the Hotelling-Downs Model.}
The Hotelling-Downs model~\cite{Hot29} originates from the problem in which two strategic ice cream vendors along a stretch of beach try to attract as many customers as possible by placing themselves. Parties only care about ``winning'' but not the ``welfare'' of their people. For the variation of the Hotelling-Downs model as such, Harrenstein et al.~\cite{HLST21} showed that computing a Nash equilibrium is 
{\sf NP}-complete in general, but it can be done in linear time when there are only two competing parties. Sabato et al.~\cite{SORR17} considered \emph{real candidacy games}, in which each agent positions itself by selecting a point from the corresponding interval on the real line and then a social choice rule determines the outcome of the competition. They established conditions for the existence of a PSNE. 

\subsubsection{On Continuous Game Convergence toward PSNE}

In this paper, the payoff function of each player is  continuous. Our work  involves algorithmic discussions on finding a Nash equilibrium of a continuous game. We briefly survey related work below. 

\paragraph{Variational Inequality Formulation.} In smooth games (i.e., each player's payoff function is continuously differentiable) with convex strategy sets, a Nash equilibrium can be characterized by a variational inequality (VI)~\cite{Sedlmayer2023,facchinei2003}. Specifically, if $F(x)$ is the concatenated pseudo-gradient of all players, then an equilibrium $x^*$ satisfies $\langle F(x^*),\,x - x^* \rangle \leq 0, \forall x\in \mathcal{S}$, 
where $\mathcal{S}$ is the product of players' feasible sets. By standard results (e.g.~\cite{Rosen1965,facchinei2003}) existence and uniqueness of equilibrium follow under appropriate monotonicity conditions, such as the diagonal strict concavity condition~\cite{Sedlmayer2023,facchinei2003}. In fact, finding an equilibrium reduces to solving the monotone inclusion $0\in F(x)+N_{\mathcal{S}}(x)$ for the normal cone~$N_{\mathcal{S}}$~\cite{Sedlmayer2023}.

\paragraph{Gradient and Extragradient Methods.} Gradient-based dynamics are a natural approach to finding equilibria. In unconstrained problems, na\"{i}ve simultaneous gradient descent-ascent (GDA) may oscillate or diverge in general continuous games~\cite{Mertikopoulos2018,Gidel2018}.  Under the VI framework, more robust methods are used, such as Gorbunov et al.'s extragradient method~\cite{Gorbunov2022a} and Popov's optimistic gradient method~\cite{popov1980,
Gorbunov2022b} are designed for monotone VIs. These methods are typically replaced by projected variants in the constrained setting.

\paragraph{Monotonicity vs.~cocoercivity.} 
A map is monotone if $\langle F(x)-F(y),x-y\rangle\ge0$ for all $x,y$ in the domain. Monotonicity holds in concave and zero-sum games~\cite{Sedlmayer2023}. Under mere monotonicity, standard analysis guarantees $O(1/T)$ convergence of averaged iterates via extragradient~\cite{Gorbunov2022a}. cocoercivity is a stronger assumption (roughly $F$ is inverse-Lipschitz) that implies Lipschitz continuity.  Intuitively, $F$ is $\beta$-cocoercive if $\langle F(x)-F(y),x-y\rangle \geq \beta\|F(x)-F(y)\|^2$,
a condition that often yields faster convergence~\cite{BauschkeCombettes2011}. Lin et al. \cite{lin2020} introduced the class of $\lambda$-cocoercive games and showed that online gradient descent (OGD) converges in the last iterate under this condition~\cite{lin2020}.

\section{The Model and Preliminaries}
\label{sec:models}

We consider a game modeling the competition of two parties that compete by proposing policies 
to society which consists of voters regarded as supporters of the two parties. 
Let~$A$ and~$B$ be the two parties which are regarded as two players in a game, 
who update their respective policies $\mathbf{z}_A, \mathbf{z}_B\in S\subset \mathbb{R}^k$ (policy vectors) iteratively, 
where $S := \{\mathbf{z}\in [-1, 1]^k: \|\mathbf{z}\|\leq 1\}$ covers the policy domain. 
The value in each dimension can be viewed as a spectrum from far-left to far-right. 
Let $\mathbf{z} := (\mathbf{z}_A, \mathbf{z}_B)$ be a state (or profile) of the game. 
Denote the supporters of~$A$ and $B$ by~$V_A$ and~$V_B$, respectively, and 
$V = V_A\dot{\cup} V_B$ be the set of all the voters where $V_A\cap V_B = \emptyset$ and $|V| = n$. We represent the \emph{preference vector} of a voter~$v\in V$ by~$\mathbf{q}_v\in S$. 
Then, for $X\in \{A, B\}$, we define the utility $u_X(\mathbf{z}_X) = \sum_{v\in V_X}\langle \mathbf{z}_X, \mathbf{q}_v\rangle$, 
which is the utility that strategy $\mathbf{z}_X$ provides to the supporters of~$X$ and~$\{A,B\}\setminus X$ respectively. 
To further simplify our discussion, we use $Q_X:= \sum_{v\in V_X}\mathbf{q}_v$ and $Q:= \sum_{v\in V}\mathbf{q}_v$. 
Hence, $u_A(\mathbf{z}_A) = \langle \mathbf{z}_A, Q_A\rangle = \mathbf{z}_A^{\top} Q_A$, 
$u_B(\mathbf{z}_A) = \langle \mathbf{z}_A, Q_B\rangle =  \mathbf{z}_A^{\top} Q_B$, 
$u_A(\mathbf{z}_B) = \langle \mathbf{z}_B, Q_A\rangle = \mathbf{z}_B^{\top} Q_A$, and 
$u_B(\mathbf{z}_B) = \langle \mathbf{z}_B, Q_B\rangle = \mathbf{z}_B^{\top} Q_B$. 
To facilitate our discussion, we assume that $\|Q_A\|, \|Q_B\|\leq 1$ as a normalization. 

\subsection{Justification of the Isotonicity of the Winning Probability}

Lin et al.~\cite{LLC2021} suggested a two-party election game model in which two parties compete in the election campaign as a two-player game by choosing a candidate, who brings utility for its supporters and the non-supporters, against each other. They raised a hypothesis that a candidate wins with higher odds if it brings more utility to all the voters. We call this monotone concept \emph{isotonicity}. 
Below, we verify such a hypothesis on the isotonicity of a winning probability.

We present a preliminary experiment designed to demonstrate an isotonic relationship between utility differences and winning probabilities. The experiment models a simplified election scenario involving two political parties, $A$ and $B$, each proposing a one-dimensional policy vector with values ranging from~$-1$ to~$1$. Consider a population of $100$ voters, each described by a one-dimensional preference vector drawn from a uniform distribution. Each voter $v$'s utility for each party's proposed policy is calculated as the inner product between the voter's preference vector and the respective party's policy vector~$z$. Let $u_v(z_A)$ and $u_v(z_B)$ represent the calculated utilities for the policies $z_A$ and $z_B$ from parties~$A$ and~$B$, respectively.

To simulate voting behavior, we apply three distinct criteria: (1) \emph{Hardmax} criterion, where party~$A$ receives a vote from~$v$ if $u_v(z_A) > u_v(z_B)$ and the vote goes to either~$A$ or~$B$ by a fair-coin flip for the tie-breaking; (2) \emph{Linear} criterion, where party~$A$ receives a vote from~$v$ with probability $\tfrac12+ (u_v(z_A) - u_v(z_B))/(2\xi)$; and (3) \emph{Softmax} criterion (or \emph{Sigmoid}), where party~$A$ receives a vote from $v$ with probability $e^{u_v(z_A)/\xi}/(e^{u_v(z_A)/\xi} + e^{u_v(z_B)/\xi})$, for a normalization factor~$\xi = 0.01$. Preference vectors for voters are drawn from a uniform distribution, for which the parameters are set with $\mu \in [-0.005,0.005]$ to fit the normalized constraint on $\|Q\| = \|Q_A+Q_B\|$, and recalculated for each trial. The simulation consists of $10,000$ independent trials to ensure statistical robustness.

From experimental results (see Figure~\ref{fig:util-prob}), we observe a clear isotonic relationship between the utility difference and the probability of winning. Specifically, the probability of party~$A$ winning tends to increase monotonically with the increasing total utility difference in favor of party~$A$. This experimental insight strongly supports the hypothesis that voters' aggregated preferences significantly influence election outcomes, and that an isotonic relationship exists between utility differences and winning probabilities. Similar results can be obtained for preference vectors drawn from a Gaussian distribution (see Appendix~A).

\begin{figure}[htbp]
        \includegraphics[width=0.475\textwidth]{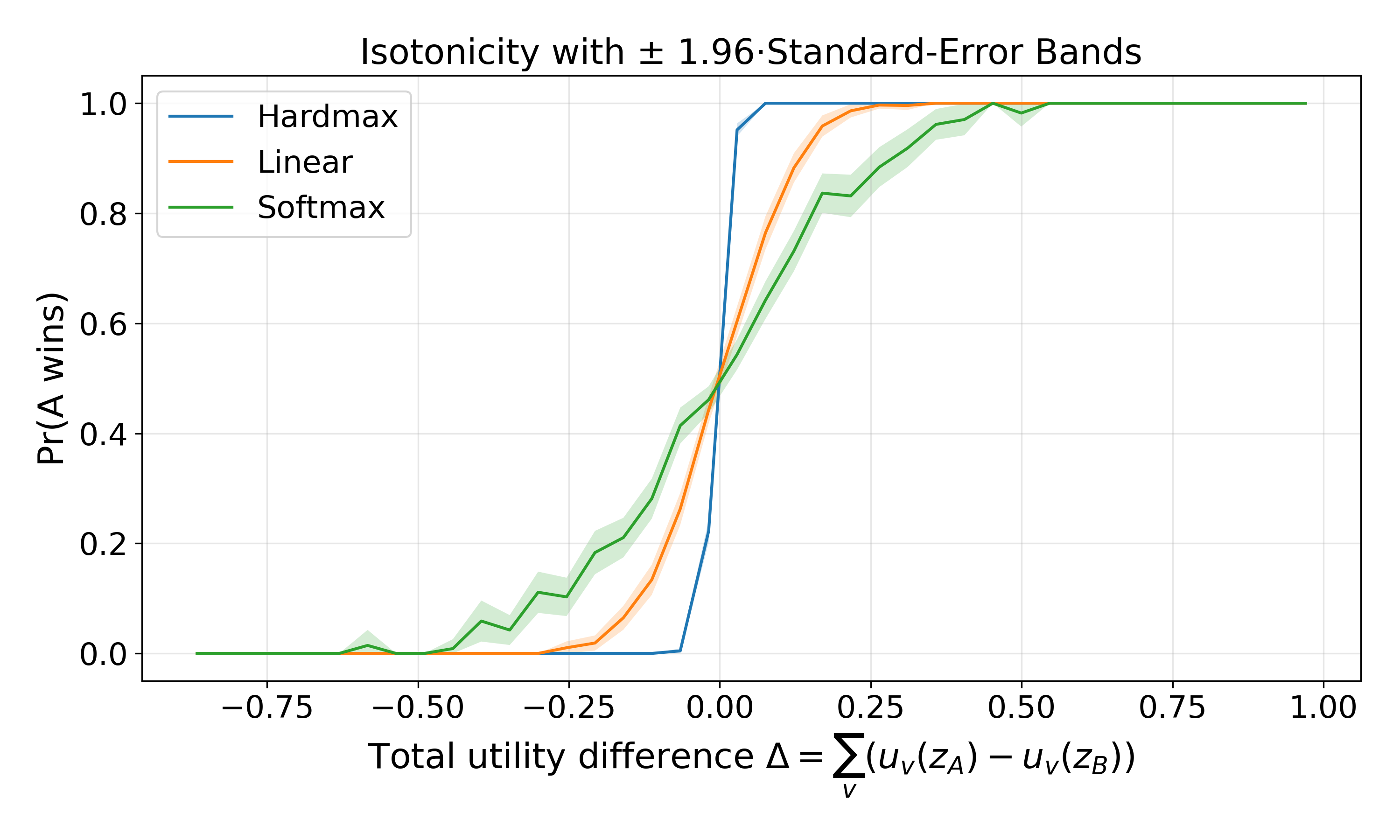}
    \caption{Utility difference vs. probability that party~$\bm{A}$ wins.}
    \label{fig:util-prob}
\end{figure}

\paragraph{A dueling-bandit perspective.} The macro perspective of the chance of winning in the election can be captured by the micro perspective of experimental 
simulations of voters' voting, and the result of the competition is determined by the utility received by the voters. These observations can be associated with the classic dueling-bandit setting in~\cite{AJK2014,du2020dueling,LLC2021}, in which we may regard a party as an arm and a better arm can only be realized with uncertainty during the competition in which the rewards of the arms are determined by the policies proposed. Hence, we model the \emph{probability of winning in a competition} as a linear link function as follows. 
Let $p_{A}$ and $p_{B} = 1- p_{A}$ denote the winning probability of party~$A$ and~$B$ respectively, where 
\begin{equation}
p_{A} = \frac{1}{2} + \frac{1}{8}(\mathbf{z}_A-\mathbf{z}_B)^{\top}Q, \;\;\,\label{eq:prob}
p_{B} = 1 - p_{A}.
\end{equation}
Here $\tfrac18$ is a normalization factor to ensure the winning probability stays in~$[0,1]$. This function maps the \emph{total utility difference}, that is, $\mathbf{z}_A^{\top}Q-\mathbf{z}_B^{\top}Q$, to a probability. Such a linear function is not only isotonic but also continuously differentiable, hence serves a good meta approximation for the payoffs of the players.

\subsection{The Payoff Functions and Their Derivatives}

Because the outcome of the competition is random, we define the \emph{payoff} functions $R_A(\mathbf{z})$ of party~$A$ 
in terms of the expected utility received by its supporters. That is, 
\begin{align*}
& R_A(\mathbf{z}) = p_{A}\cdot 
\mathbf{z}_A^{\top} Q_A + (1-p_{A})\cdot \mathbf{z}_B^{\top} Q_A\\
& = \tfrac12(\mathbf{z}_A^{\top}Q_A+\mathbf{z}_B^{\top}Q_A) + \tfrac18(\mathbf{z}_A^{\top}Q-\mathbf{z}_B^{\top}Q)(\mathbf{z}_A^{\top}Q_A-\mathbf{z}_B^{\top}Q_A).
\end{align*}
$R_B(\mathbf{z})$ can be defined similarly. 
For $\mathbf{z}_A, \mathbf{z}_B\in S$, the maximizer of $R_A(\mathbf{z})$ and $R_B(\mathbf{z})$ can be computed in a decentralized way, that is, 
$\argmax_{\mathbf{z}_A\in S}R_A(\mathbf{z})$ and $\argmax_{\mathbf{z}_B\in S}R_B(\mathbf{z})$. The gradient of $R_A(\mathbf{z})$ w.r.t.~$\mathbf{z}_A$ can be derived as 
\[
\pdv{R_A(\mathbf{z})}{\mathbf{z}_A} 
= \tfrac12 Q_A + \tfrac18(\mathbf{z}_A^{\top}Q-\mathbf{z}_B^{\top}Q) Q_A + \tfrac18(\mathbf{z}_A^{\top} Q_A-\mathbf{z}_B^{\top} Q_A) Q.
\] $\pdv{R_B(\mathbf{z})}{\mathbf{z}_B}$ can be obtained similarly. 
Hence, we derive the second-order derivatives 
\begin{align*}
    &\frac{\partial^2 R_X(\mathbf{z})}{\partial\mathbf{z}_X^2}[i,j] = \frac{1}{8}\left(Q[i]Q_X[j]+Q[j]Q_X[i]\right), 
\end{align*}
for $X\in \{A, B\}$. Thus, the payoff function 
is not necessarily convex or concave, as seen from the second-order derivatives.

Without loss of generality, we restrict attention to profiles satisfying the rationality condition called \emph{egoistic property}, defined below.

\begin{defn}[egoistic property]\label{defn:egoism}
We say that the state $\mathbf{z}$ of the two-party policy competition game satisfies the egoistic property if $\mathbf{z}_A^{\top} Q_A\geq \mathbf{z}_B^{\top} Q_A$ and $\mathbf{z}_B^{\top} Q_B\geq \mathbf{z}_A^{\top} Q_B$. That is, each party's policy provides no less utility to its own supporters than the opponent's policy. 
\end{defn}

\paragraph*{Remark.} Suppose contrarily that $\mathbf{z}_B^{\top} Q_A > \mathbf{z}_A^{\top}Q_A$, then party~$A$ could switch to the opponent's policy ~$\mathbf{z}_B$, which would only increase party~$A$'s payoff. Indeed, by such a switch the winning probability of party~$A$ becomes $1/2$, increases its payoff by~$\mathbf{z}_B^{\top}Q_A - (p_{A}\cdot\mathbf{z}_A^{\top}Q_A + p_{B}\cdot\mathbf{z}_B^{\top}Q_A)\geq \mathbf{z}_B^{\top}Q_A - (p_{A}\cdot\mathbf{z}_B^{\top}Q_A + p_{B}\cdot\mathbf{z}_B^{\top}Q_A) = 0$ since $p_{A}+p_{B} = 1$. Hence, egoistic property must hold in a PSNE. 

\section{The Existence of PSNE in the One-Dimensional Setting}
\label{sec:results_I}

For $k=1$, the policies $z_A, z_B$ lie at $[-1, 1]\subset \mathbb{R}$. For party~$A$, we have 
    $R_A(\mathbf{z}) = \tfrac12(z_A+z_B)Q_A + \tfrac18QQ_A(z_A-z_B)^2$,
where $\mathbf{z} = (z_A, z_B)$ denotes the profile (or state) of the game. 
Then, we have \[
    \odv{R_A(\mathbf{z})}{z_A} = \tfrac12 Q_A + \tfrac14QQ_A(z_A-z_B),\;\; 
    \odv[order={2}]{R_A(\mathbf{z})}{z_A} = \tfrac14QQ_A.\]

\begin{thm}\label{thm:degenerate_NE}
The two-party policy competition game in the one-dimensional setting has a PSNE in a closed form.  
\end{thm}
\begin{proof}
Consider the two-party policy competition game in the one-dimensional setting. WLOG, below we consider the cases w.r.t. party~$A$ and those w.r.t. party~$B$ can be obtained similarly. 
\begin{enumerate}
    \item [(1)] If $QQ_A\geq 0$, then $R_A(\mathbf{z})$ is convex since $\odv[order={2}]{R_A(\mathbf{z})}{z_A}\geq 0$. Assume that $Q\geq 0$ and $Q_A\geq 0$. The egoistic property guarantees that $z_A Q_A\geq z_B Q_A$ and hence we have $\odv{R_A(\mathbf{z})}{z_A}\geq 0$. Therefore, we derive that setting the policy $z_A = 1$ maximizes the payoff of party~$A$. Next, consider the case that $Q\leq 0$ and $Q_A\leq 0$. The egoistic property leads to the fact that $Q_A(z_A-z_B)\geq 0$. Since $Q_A\leq 0$, we have $z_A\leq z_B$. Moreover, it is clear that $QQ_A(z_A-z_B)\leq 0$ and hence $\odv{R_A(z)}{z_A}\leq 0$. Therefore, setting $z_A = -1$ maximizes the payoff of party~$A$.
    \item [(2)] If $QQ_A < 0$, then $R_A(\mathbf{z})$ is concave since $\odv[order={2}]{R_A(\mathbf{z})}{z_A} < 0$. Assume that $Q > 0$ and $Q_A < 0$. It is clear that we must have $Q_B > 0$ otherwise $Q$ can never be positive. By the arguments in~(1), we know that $z_B = 1$. Since $R_A(\mathbf{z})$ w.r.t.~$z_A$, we then focus on solving $\odv{R_A(z)}{z_A} = 0$, which implies that $z_A = -\frac{1}{2}Q_A\left(\frac{4}{QQ_A}\right)+1 = 1-\frac{2}{Q}$.  
    Fit $z_A$ to the feasible space $S = [-1, 1]$, we derive $z_B = 1$ and \[
    z_A = \left\{\begin{array}{ll}
    -1 & \mbox{ if } 0 < Q < 1\\
    1-\frac{2}{Q} & \mbox{ if } 1\leq Q
    \end{array}\right..
    \] 
    Similarly, for the case that $Q < 0$ and $Q_A > 0$, we obtain $z_B = -1$ and 
    \[
    z_A = \left\{\begin{array}{ll}
    1 & \mbox{ if } -1\leq Q < 0\\
    -1-\frac{2}{Q} & \mbox{ if } Q \leq -1
    \end{array}\right..
    \] 
\end{enumerate}
\end{proof}

\section{The Existence of PSNE in the Multi-Dimensional Setting}
\label{sec:PSNE_general_case}

In this section,  we consider the general case that $k\geq 1$. 
We claim that it is sufficient for party~$A$ to consider the space $\mbox{span}(\{Q_A, Q_B\}) := \{r_1Q_A+r_2Q_B: r_1, r_2\in \mathbb{R}\}$. Indeed, for a policy $\mathbf{z}_A\in S\subset \mathbb{R}^k$ 
which is not in~$\mbox{span}(\{Q_A, Q_B\})$, there must exist a vector $\mathbf{v}\in S$ such that $\mathbf{z}_A = \hat{\mathbf{z}}_A + \mathbf{v}$ where $\hat{\mathbf{z}}_A\in \mbox{span}(\{Q_A, Q_B\})$ and $\langle \mathbf{v}, \mathbf{n}\rangle > 0$ for a normal vector $\mathbf{n}$ perpendicular to any vector in~$\mbox{span}(\{Q_A, Q_B\})$. Denote by $\mathbf{v}_{\mathbf{n}}$ the vector parallel to $\mathbf{n}$. Clearly, $\mathbf{v}_{\mathbf{n}}$ will not affect the payoff~$R_A(\mathbf{z})$.

\begin{fact}\label{fact:inner_product}
For two vectors $\mathbf{u},\mathbf{v}\in \mathbb{R}^k$, $k\geq 1$, if the dot product is applied, then $\langle \mathbf{u}, \mathbf{v}\rangle = \|\mathbf{u}\|\|\mathbf{v}\|\cos(\theta)$,
where $\theta$ is the angle between $\mathbf{u}$ and $\mathbf{v}$ and $\|\cdot\|$ represents the 2-norm.
\end{fact}

As we can focus on $\mathbf{z}_A,\mathbf{z}_B\in \mbox{span}(\{Q_A, Q_B\})$, $\mathbf{z}_A$ and $\mathbf{z}_B$ can be represented by the polar coordinates as $(r_A, \theta_A)$ and $(r_B, \theta_B)$, respectively, where $r_A = \|\mathbf{z}_A\|, r_B = \|\mathbf{z}_B\|$, respectively, and $\theta_A, \theta_B>0$ are the angles between~$Q_A$ and~$\mathbf{z}_A$ and the angles between~$Q_B$ and~$\mathbf{z}_B$, respectively. 
Then, we have $$p_A = \tfrac12+\tfrac18\left(r_A\|Q\|\cos(\rho_A-\theta_A) - r_B\|Q\|\cos(\rho_B-\theta_B)\right),$$  and $R_A(\mathbf{z}) = p_A(r_A\|Q_A\|\cos(\theta_A))+(1-p_A)(r_B\|Q_A\|\cos(\rho_A+\rho_B-\theta_B))$. 
Note that $(r_A, \theta_A)$ can be transformed back to $\mathbf{z}_A\in S$ by multiplying a rotation matrix 
to~$Q_A$. 

\paragraph{Remark.} We observe that the variables in the payoff function $R_A(\mathbf{z})$ are of $k$-dimensional, and hence the gradients and second-order derivatives of~$R_A(\mathbf{z})$ are vectors and a Hessian matrix, respectively. In polar coordinates, one may reduce to considering only two variables, i.e., $r_A$ and $\theta_A$, which significantly simplifies the analysis.  

\subsection{Consensus-Reachable Case}
\label{subsec:consensus-reachable}

Before stepping into the following discussions, let us define the following property which will be considered as an assumption in this section. 
\begin{defn}\label{defn:consensus}
    A two-party policy competition game is \emph{consensus-reachable} if $Q_A^{\top}Q\geq 0$ and $Q_B^{\top}Q\geq 0$.
\end{defn}
The term ``consensus-reachable" comes from the fact that the consensus $Q = Q_A+Q_B$ has positive cosine similarity between it and~$Q_A$ or~$Q_B$. Hence, each party's supporters have preferences that are not too ``far-apart''.
Under this assumption, 
any policy of each party brings nonnegative utility not only to their own supporters but also to all the voters.

\begin{lem}\label{lem:consensus-norm}
If the two-party policy competition game is consensus-reachable, then the two players must choose $r_A = r_B = 1$ to maximize their payoffs.    
\end{lem}
\begin{proof} (sketch) 
The second derivative of the payoff function $R_A(\mathbf{r}, \bm{\theta})$ w.r.t. $r_A$ is 
$\pdv[order=2]{R_A(\mathbf{r}, \bm{\theta})}{r_A} = \tfrac12\|Q_A\|\|Q\| \cos(\theta_A)\cos(\rho_A-\theta_A)$,  
which is nonnegative by the consensus-reachable property. Thus, $R_A(\mathbf{r}, \bm{\theta})$ is convex w.r.t.~$r_A$, and hence the maximizer $r_A$ is at the boundary, which is either at~1 or~0. Next, we can derive  $R_A((1, r_B), \bm{\theta}) - R_A((0, r_B), \bm{\theta}) = \frac{1}{2}\|Q_A\|\cos(\theta_A) + \tfrac14\Bigl(\|Q\|\|Q_A\|\cos(\rho_A-\theta_A)\cos(\theta_A)\Bigr) -\tfrac14\Bigl(\mathbf{z}_B^{\top}Q\|Q_A\|\cos(\theta_A) +\mathbf{z}_B^{\top}Q_A\|Q\|\cos(\rho_A-\theta_A)\Bigr)\geq 0$ using the Cauchy-Schwarz inequality and the egoistic property. Thus, choosing $r_A = 1$ maximizes party~$A$'s payoff. It similarly holds for party~$B$. 
Refer to Appendix~B 
for the full proof.
\end{proof}

\paragraph{Remark.} Intuitively, since the game is consensus-reachable, maximizing~$r_A$ increases the winning probability of party~$A$ and the utility $u_A(\mathbf{z}_A)$ as well. Lemma~\ref{lem:consensus-norm} serves as a thorough validation of this intuition.

Next, we argue that 
considering $\theta_A\in [0, \rho_A]$ and $\theta_B\in [0, \rho_B]$ is sufficient. 
Indeed, whenever there is a strategy $\mathbf{z}_A'$ with the angle (i.e., $\theta_A'-\rho_A$) between $Q$ and $\mathbf{z}_A'$, one can find a $\mathbf{z}_A$ in $[0, \rho_A]$ with the same angle $\rho_A-\theta_A = \theta_A'-\rho_A$. Hence, the winning probability remains the same while the utility $u_A(\mathbf{z}_A)$ is larger. Similarly, a strategy $\mathbf{z}_A''$ with $\theta_A'' = -\theta_A$ can be replaced by $\mathbf{z}_A$ since the winning probability increases (see Figure~\ref{fig:policies_in_alpha}). 

\begin{figure}[ht]
    \centering
    \includegraphics[width=0.99\columnwidth]{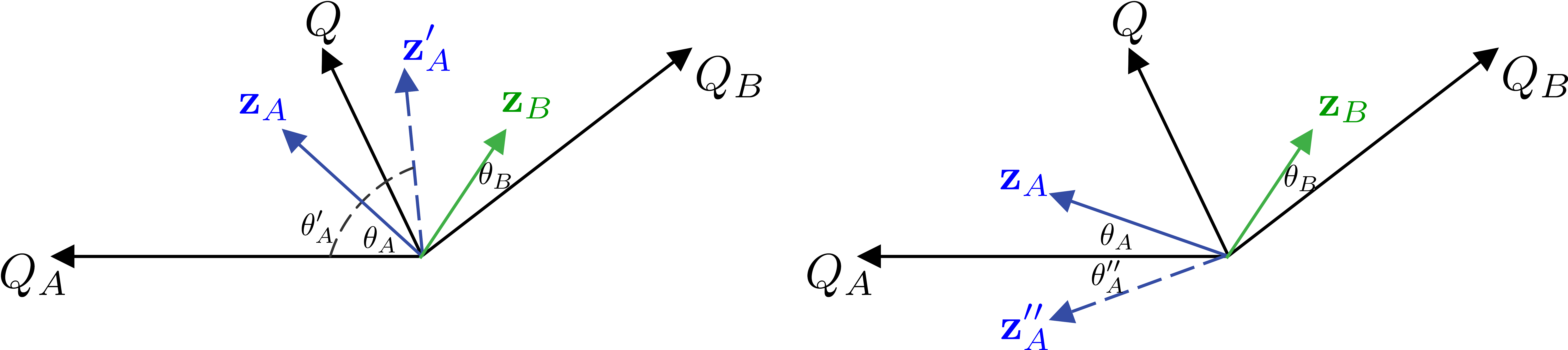}
    \caption{The reasonable range of a policy's angle.}
    \label{fig:policies_in_alpha}
\end{figure}

\paragraph{Remark.} As the game is consensus-reachable, the policy of party~$B$ can be chosen similarly. We prove that the payoff functions of party~$A$ and~$B$, provided with $r_A = r_B = 1$ by Lemma~\ref{lem:consensus-norm}, are concave with respect to~$\theta_A$ and~$\theta_B$. By Kakutani's fixed point theorem, with the convexity and compactness of the domain and continuity of the payoff functions, we have Theorem~\ref{thm:result_consensus} to conclude this subsection (see the detailed proof in Appendix~C). 

\begin{thm}\label{thm:result_consensus}
The two-party policy competition game always has a PSNE if it is consensus-reachable.  
\end{thm}
\subsection{The Non-Consensus-Reachable Case}
\label{subsec:non-consensus-reachable}

Without the consensus-reachable condition, 
it is unclear whether previous analysis still holds. For example, when $\rho_A-\theta_A > \pi/2$, increasing $r_A$ harms the winning probability $p_A$ 
because $\mathbf{z}_A^{\top} Q$ is negative. A party's policy which is ``good'' for $Q_A$ 
is `'NOT necessarily good'' for $Q$. 

In this section, we consider the case that $\rho_A > \pi/2$ and $\rho_B \leq \pi/2$. The other case that $\rho_B > \pi/2$ and $\rho_A \leq \pi/2$ can be handled similarly. Since the angle between $Q_A$ and $Q_B$ never exceeds $\pi$, at most one of $\rho_A$ and $\rho_B$ is greater than $\pi/2$. 

\subsubsection{$\bm{r_A}$ and $\bm{r_B}$.}

Firstly, we can prove that $\pdv{R_A(\mathbf{r}, \bm{\theta})}{r_A}$ is always nonnegative for $\theta_A\in [0,\pi/2]$, so for $\theta_A$ is this range choosing $r_A = 1$ is the maximizer. For $\theta_A\in [\pi/2, \rho_A]$, we can also derive that $\pdv{R_A(\mathbf{r}, \bm{\theta})}{r_A}\big|_{\theta_A=\pi/2}\geq 0$ and $\pdv[order=2]{R_A(\mathbf{r}, \bm{\theta})}{r_A} \leq 0$ for $\theta_A\geq \pi/2$. So, either $\pdv{R_A(\mathbf{r}, \bm{\theta})}{r_A}\geq 0$ for all $\theta_A\in [\pi/2, \rho_A]$, or there exists $\theta_A'\in [\pi/2,  \rho_A]$ such that 
$\pdv{R_A(\mathbf{r}, \bm{\theta})}{r_A}< 0$ for $\theta_A\in (\theta_A', \rho_A]$. The latter implies that $r_A = 0$ maximizes $R_A$, however, one can replace the policy of $r_A=0$ with another better response, say, $r_A = 1$ and $\theta_A = \pi/2$. Hence, choosing~$r_A~=~1$ always maximizes party~$A$'s payoff. For party~$B$, similar to the analysis for the consensus-reachable case, choosing $r_B = 1$ maximizes party~$B$'s payoff. Hence, we obtain Lemma~\ref{lem:non-consensus-reachable-norm-maximum} (refer to Appendix~D 
for more detail).

\begin{lem}\label{lem:non-consensus-reachable-norm-maximum}
For the nonconsensus-reachable case where $\rho_A > \pi/2$ or $\rho_B > \pi/2$, 
the best responses of the two players always set $r_A = r_B = 1$.    
\end{lem}

Under the condition that $r_A = r_B = 1$ is fixed, we can show that the payoff functions of party~$A$ and party~$B$, in terms of functions of variables $\theta_A$ and $\theta_B$, respectively, are concave--quasi-concave (see the proof of Theorem~\ref{thm:non-consensus-reachable-PSNE}). Hence, for the non-consensus-reachable case of the game, 
we can guarantee that a PSNE exists. 

\begin{thm}\label{thm:non-consensus-reachable-PSNE}
The two-party policy competition game has at least one PSNE even it is not consensus-reachable.
\end{thm}
\begin{proof}
(sketch) 
Since $r_A=r_B = 1$ is the condition for the best responses of the two players, it suffices to investigate the strategies in terms of the ``angles''. To further simplify our discussion, we replace $\cos(\theta_A)$ and $\cos(\theta_B)$ by~$x$ and $y$ respectively, and let $f(x):= R_A(r_A=1, \theta_A)$ and $g(y):= R_B(r_B=1, \theta_B)$. Since $\theta_A, \theta_B\in [0,\pi/2]$, it is clear that $x, y\in [0, 1]$. Then, we rewrite $f(x)$ and $g(y)$ as  
\begin{align*}
& f(x) = 
\bigl(\tfrac12 + D_0\bigl(C_1\,x + \sqrt{1 - C_1^2}\,\sqrt{1 - x^2} - C_3\bigr)\bigr)\,D_1\,x \\
& \;+\;
\bigl(\tfrac12 - D_0\bigl(C_1\,x + \sqrt{1 - C_1^2}\,\sqrt{1 - x^2} - C_3\bigr)\bigr)\,D_1\,C_4,\\
& g(y) = 
\bigl(\tfrac12 + D_0\bigl(C_2\,y + \sqrt{1 - C_2^2}\,\sqrt{1 - y^2} - C_3'\bigr)\bigr)\,D_2\,y \\
& \;+\;
\bigl(\tfrac12 - D_0\bigl(C_2\,y + \sqrt{1 - C_2^2}\,\sqrt{1 - y^2} - C_3'\bigr)\bigr)\,D_2\,C_4',
\end{align*}
where 
$D_0:= \|Q\|/8$, $D_1: = \|Q_A\|$, $D_2 := \|Q_B\|$, $C_1 := \cos\rho_A$, $C_2:= \cos\rho_B$, $C_3 := \cos(\rho_B-\theta_B)$, $C_3' := \cos(\rho_A-\theta_A)$, $C_4:= \cos(\rho_A+\rho_B-\theta_B)$, and $C_4' := \cos(\rho_A+\rho_B-\theta_A)$. The theorem can be proved by arguing $f(x)$ is concave and $g(y)$ is quasi-concave in the respective domains (refer to Appendix~E 
for the details). 
\end{proof}

\paragraph{Remark.} From Theorem~\ref{thm:degenerate_NE}, Lemma~\ref{lem:consensus-norm} and~\ref{lem:non-consensus-reachable-norm-maximum}, we observe that 
at any PSNE of the two-party policy competition game, both parties choose policies with \emph{maximal strength}. Equilibria arise via \emph{directional} (angular) adjustments toward \(Q\), not via \emph{intensity} compromise. Relative to median-voter models, this predicts \emph{max-intensity} platforms with \emph{directional moderation}, capturing one facet of extreme politics.

\section{Bad and Good News for the Gradient-Based Approach}
\label{sec:nonmonotone_nonLipschitz}

First, we reiterate the definition of pseudo-gradient map of an $n$-player game defined by Rosen~\cite{Rosen1965} so that the discussions will be self-contained.

\paragraph{Remark} Throughout this section we use the pseudo-gradient of payoffs, $F$ is the concatenated pseudo-gradients of players' \emph{payoffs}, not their costs. Consequently, the associated VI is written with ``$\leq$''. One can equivalently work with $-F$ and the standard ``$\geq$'' VI convention.   

\begin{defn}[Pseudo‐Gradient~\cite{Rosen1965}]
\label{defn:pseudo-gradient}
The \emph{pseudo‐gradient map} \(F : S \to \prod_{i=1}^n \mathbb{R}^{d_i}\) is the concatenation of each player’s own payoff‐gradient with respect to its strategy 
$F(\mathbf{x}) = (\nabla_{x_i} u_i(x_i, \mathbf{x}_{-i}))_{i=1}^n$, 
where $u_i$ is the differentiable payoff function of player~$i$ and $\mathbf{x} = (x_1, x_2, \ldots, x_n)\in S$ is the strategy profile.
\end{defn}
Note that a profile $\mathbf{x}^*$ is a PSNE 
if and only if the \emph{variational inequality} $F(\mathbf{x}^*)^{\top}(\mathbf{y}-\mathbf{x}^*)\leq 0$ for any $\mathbf{y}\in S$. 

\paragraph{Remark.} Since cocoercivity implies monotonicity \(\langle F(u)-F(v),u-v\rangle\leq 0\),
failure of monotonicity on \([0,1]^2\) immediately disproves any \(\lambda\)-cocoercivity.

We revisit the general case that policies are in $[-1, 1]^k\subset \mathbb{R}^k$ and can be transformed into the polar coordinates with respect to $Q_A$ and $Q_B$. 
As discussed in Sect.~\ref{sec:PSNE_general_case}, we simply assume $r_A = r_B = 1$ and 
denote $r_A\cos(\theta_A)$ and $r_B\cos(\theta_B)$ by $x$ and $y$, respectively. 

\subsection{A Counterexample of Monotonicity}

Define that $F(x,y) \;=\;(F_{1}(x,y),\,F_{2}(x,y))
\;=\;\bigl(\nabla_x f_x(x,y),\;\nabla_y g_y(x,y)\bigr)$,
in which $x:=\cos(\theta_A), y:=\cos(\theta_B)$, $f_x:= R_A(\theta_A)$ and $f_y:= R_B(\theta_B)$ for $r_A = r_B = 1$ as previously defined. 
Specifically for further discussions, recall that $p_{A} =\frac12+\frac{\|Q\|}{8}(C_{1}x +\sqrt{1-C_{1}^{2}}\;\sqrt{1-x^{2}} - C_{2}y -\sqrt{1-C_{2}^{2}}\;\sqrt{1-y^{2}})$, and we have 
\begin{align*}
& f_{x}(x,y) =\|Q_{A}\|\Bigl(p_{A} x + (1-p_{A})(Ky+L\sqrt{1-y^2})\Bigr), \\
& g_{y}(x,y) =\|Q_{B}\|\Bigl((1-p_{A})\,y + p_{A}(Kx+L\sqrt{1-x^2})\Bigr),
\end{align*}
where $C_1= \cos(\rho_A), C_2= \cos(\rho_B)$, $K = C_{1}C_{2}-\sqrt{1-C_{1}^{2}} \sqrt{1-C_{2}^{2}}$, and $L = \sqrt{1-C_{1}^{2}} C_{2} + C_{1}\,\sqrt{1-C_{2}^{2}}$.
The pseudo‐gradient components can be derived as 
\begin{equation*}\label{eq:pseudogradient}
\begin{aligned}
\nabla_{x}f_{x}(x,y)
&= \|Q_{A}\|\Bigl(\,p_{A}
+ \bigl(x - M(y)\bigr)\,\frac{\partial p_{A}}{\partial x}\Bigr),\\[1ex]
\nabla_{y}g_{y}(x,y)
&= \|Q_{B}\|\Bigl((1 - p_{A})
+ \bigl(M(x) - y\bigr)\,\frac{\partial p_{A}}{\partial y}\Bigr),
\end{aligned}
\end{equation*}
where $M(t) = K\,t + L\sqrt{1 - t^{2}}$.
%
Let $\|Q_A\| = \|Q_B\| = 1$ and $\rho_A = \rho_B$ such that $\cos(\rho_A) = \cos(\rho_B) = 0.6$. Since the angle between $Q_A$ and $Q_B$ is $2\rho_A$, the law of cosines implies that $\|Q\| = (2\|Q_A\|^2-2\|Q_A\|^2\cos(\pi-2\rho_A))^{1/2} = 1.2$.  
Then, by choosing $\mathbf{v}_1 = (x_{1},y_{1}) = (0.40,0)$ and $\mathbf{v}_2 = (x_{2},y_{2}) = (0, 0.43)\), we derive that $F(\mathbf{v}_1) \approx (0.5049, 0.4049)$ and $F(\mathbf{v}_2)\approx (0.4058, 0.5096)$, hence $F(\mathbf{v}_1)-F(\mathbf{v}_2)\approx (0.0991, -0.1047)$ and $\mathbf{v}_1 - \mathbf{v}_2 = (0.40, -0.43)$. 
We obtain that $(F(z_1) - F(z_2))^{\top}(z_1-z_2)\approx 0.0847 > 0$. 
Let $\bm{\theta_i} = (\arccos(x_i), \arccos(y_i))$ for $i=1,2$. 
Note that by $\odif(\cos\theta)/\odif(\theta) = -\sin(\theta)$, the pseudo-gradient $\tilde{F}$ with respect to~$\theta_A$ and~$\theta_B$ is $(F_1(x, y)\cdot (-\sqrt{1-x^2}), F_2(x, y)\cdot (-\sqrt{1-y^2}))$, we still have $(\tilde{F}(\bm{\theta}_1) - \tilde{F}(\bm{\theta}_2))^{\top}(\bm{\theta}_1-\bm{\theta}_2) 
< 0$. 
Hence, this is a counterexample of monotonicity which also disproves the $\lambda$-cocoercivity of the game for any~\(\lambda>0\). 

\subsection{Decentralized Gradient Ascent Experiments}

We have shown that the two-party policy competition game is not necessarily a convex or concave game, and the game is not guaranteed to be monotone either. Nevertheless, 
we would like to know to what extent a gradient-based algorithm can help in our game. Let us consider a vanilla  gradient ascent algorithm by which the two players update their policies $\mathbf{z}_A^{(t)}$ and $\mathbf{z}_B^{(t)}$ in parallel from $\mathbf{z}_A^{(t-1)}$ and $\mathbf{z}_B^{(t-1)}$, respectively, in iterations:  
\begin{align*}
&\mathbf{z}_X^{(t)}\gets \Pi_{\mathcal{S}}\Bigl(\mathbf{z}_X^{(t-1)} + \eta_t\frac{\partial R_X(\mathbf{z}^{(t-1)})}{\partial \mathbf{z}_X}\Bigr), \;\mbox{ for } X\in \{A, B\},
\end{align*}
where $\eta_t = 1/t^{0.75}$ is the time-decaying step size chosen according to Robbins-Monro conditions~\cite{robbins1951,mertikopoulos2024}\footnote{Robbins-Monro conditions: $\sum_{t\geq 1}\eta_t = \infty$,\, $\sum_{t\geq 1}\eta_t^2 < \infty$.}. 
The projection $\Pi_{\mathcal{S}}$ is done by considering two scenarios. First, if the updated policy has norm greater than~$1$, divide it by its norm. Second, if, say $\mathbf{z}_A$, is outside the ``wedge'' spanned by~$Q_A$ and~$Q$, we perform the reflection of $\mathbf{z}_A$ across the nearer boundary $Q'\in\{Q_A, Q\}$ 
which flips the signed angle of $\mathbf{z}_A$ across~$Q_A$ and maintains its norm. Such a reflection projects $\mathbf{z}_A$ onto~$S$ with either the same $u_A(\mathbf{z}_A)$ but higher $p_A$ or the same $p_A$ but larger~$u_A(\mathbf{z}_A)$. 

\begin{figure}[ht]
    \centering
    \begin{subfigure}[b]{0.475\textwidth}
    \includegraphics[width=0.9\columnwidth]{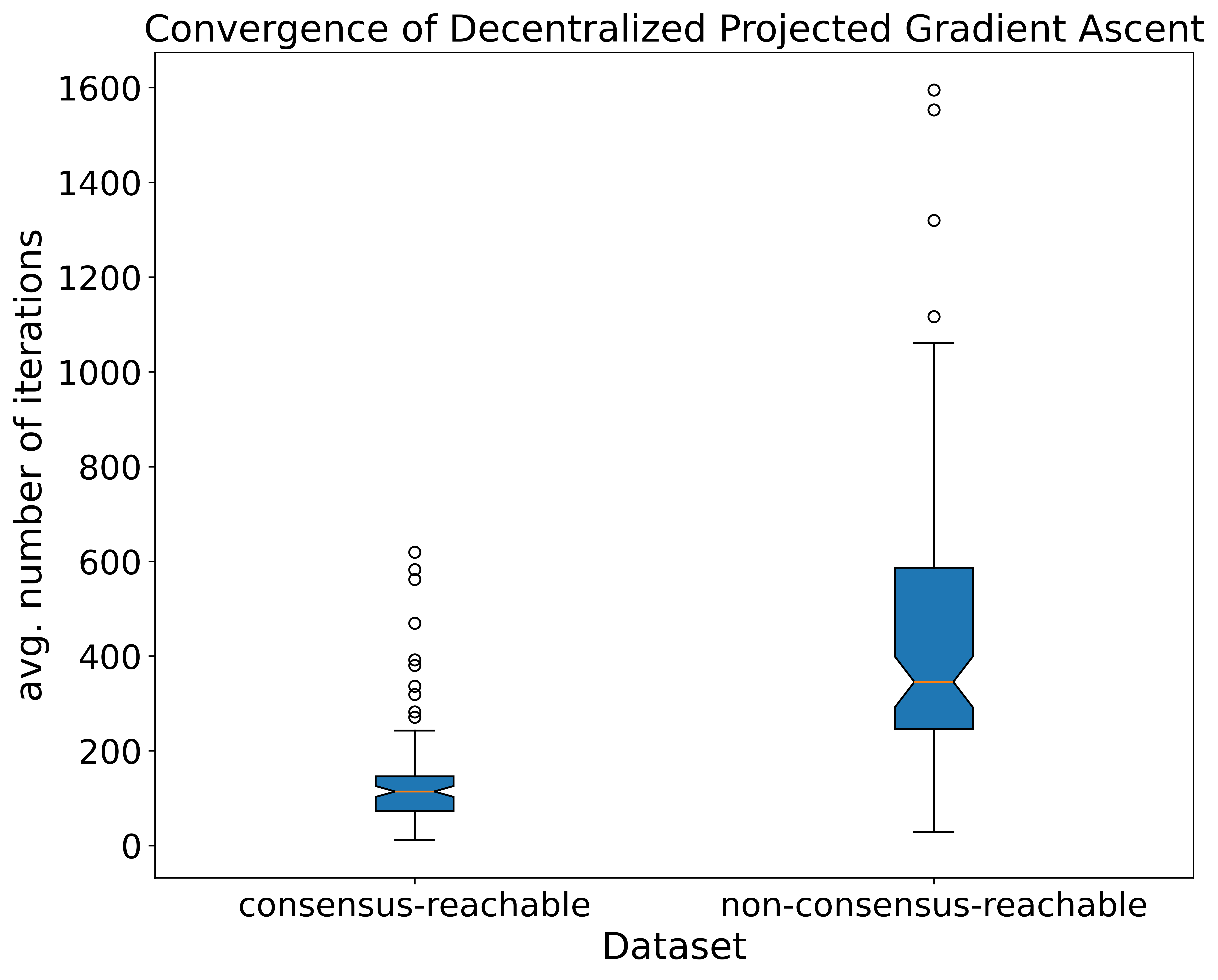}
    \caption{Average number of iterations until convergence.}
    \label{fig:convergence}
    \end{subfigure}
    \hfill
    \begin{subfigure}[b]{0.475\textwidth}
    \includegraphics[width=0.9\columnwidth]{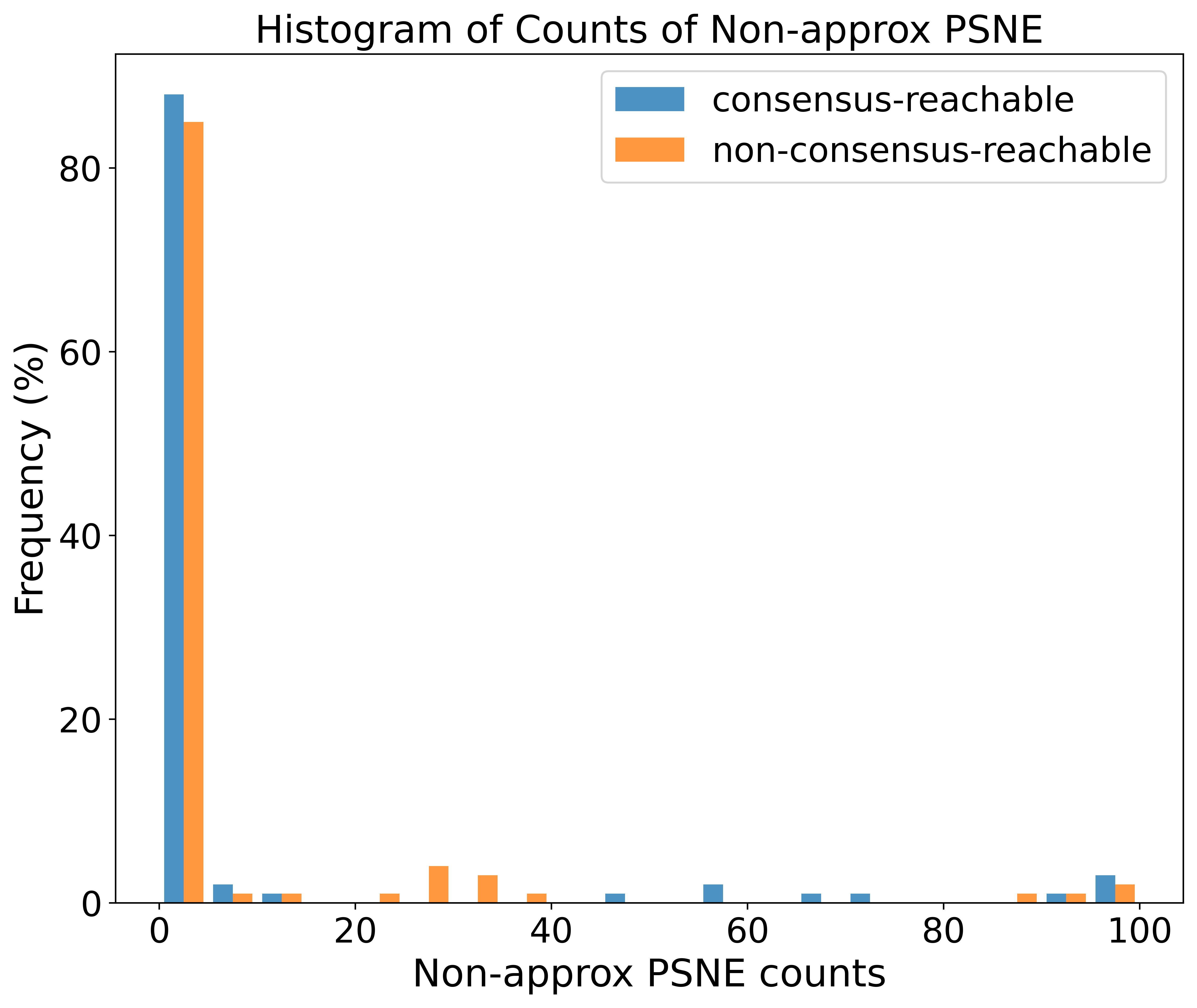}   
    \caption{Frequency of convergence that is not an approx. PSNE.}
    \label{fig:nonNE_histogram}
    \end{subfigure}
    \caption{Convergence analysis of decentralized projected gradient ascent. (a) 
    The average number of iterations until convergence for 100 random starting policy profiles for each instance. 
    (b) The frequency of the decentralized projected gradient ascent not converging to an approximate PSNE.} 
    \label{fig:GA_experiments}
\end{figure}

We prepare 100 random instances of voters' preferences~$Q_A$ and~$Q_B$, each of them is collectively considered as a vector of norm bounded by~$1$ in~$\mathcal{S}$ and is uniformly sampled. The decentralized projected gradient ascent is run from scratch for 100 times, in each time two random initial policies $\mathbf{z}_A$ and $\mathbf{z}_B$ are uniformly sampled from the intersection of $\mathcal{S}$ and the wedge spanned by~$Q_A$ and~$Q$ and the wedge spanned by~$Q_B$ and~$Q$, respectively. The algorithm runs in parallel for each player 
with the convergence condition which is either (i) the maximum number of iterations is reached or (ii) the policy differences satisfy $\max\{\|\mathbf{z}_A^{(t)} -\mathbf{z}_A^{(t')}\|, \|\mathbf{z}_B^{(t)} -\mathbf{z}_B^{(t')}\|\}\leq 10^{-4}$ for some $t' < t$. When the convergence 
is reached, we check whether the policy profile is an (approximate) PSNE by comparing the payoff of unilateral deviation to any grid-policy in~$\mathcal{G} = \{\mathbf{s}\in G \times G: \|\mathbf{s}\|\leq 1\mbox{ and }\mathbf{s}\mbox{ is in the wedge spanned by } Q_A, Q_B\}$, for $G = \{-1+0.1k: 0\leq k\leq 20\}$. We then compute the frequency that 
an initial profile 
leads to a convergent state under which a player can still gain more payoff by unilaterally deviating to a grid-policy. A global picture of our experiments is illustrated in Figure~\ref{fig:GA_experiments}.

The experimental results justify that a gradient-based algorithm neither always leads to an (approximate) PSNE nor converges to a unique profile (see Figure~\ref{fig:GA_experiments_multi-convergence}); nevertheless, all the runs of the projected gradient ascent converge within 4,000 iterations---most in just a few hundred. Figure~\ref{fig:convergence} shows that the algorithm usually converges faster in the consensus-reachable case than in the non-consensus-reachable case ($p$-value $<0.001$ by the Wilcoxon signed-rank test). 
Figure~\ref{fig:nonNE_histogram} shows that the algorithm converges to an approximate PSNE in the vast majority of cases (92.8\% and 93.4\% for a total of 10,000 runs in the consensus-reachable and non-consensus-reachable cases, respectively), though the convergence to an approximate PSNE does not significantly differ in the two cases by Fisher's exact test. The simplicity of implementation and moderately rapid convergence to PSNE profiles in most cases 
make the algorithm an attractive and practical solution. 
Refer to Appendix~G 
for more discussion and Appendix~I 
for the computing environment setting. 

As a robustness check, we also implemented a decentralized projected extra-gradient ascent update which has been studied for its faster convergence in equilibrium learning (see~\cite{korpelevich1976extragradient,nemirovski2004prox,daskalakis2017training,mokhtari2020unified}). However, across both datasets, the distributions of non-approximate PSNE counts and iterations-to-approximate-PSNE are visually similar to the vanilla one. We do not observe a consistent improvement (See Appendix~H).

\begin{figure}[h]
    \centering
    \begin{subfigure}[b]{0.45\textwidth}
    \includegraphics[width=0.85\textwidth]{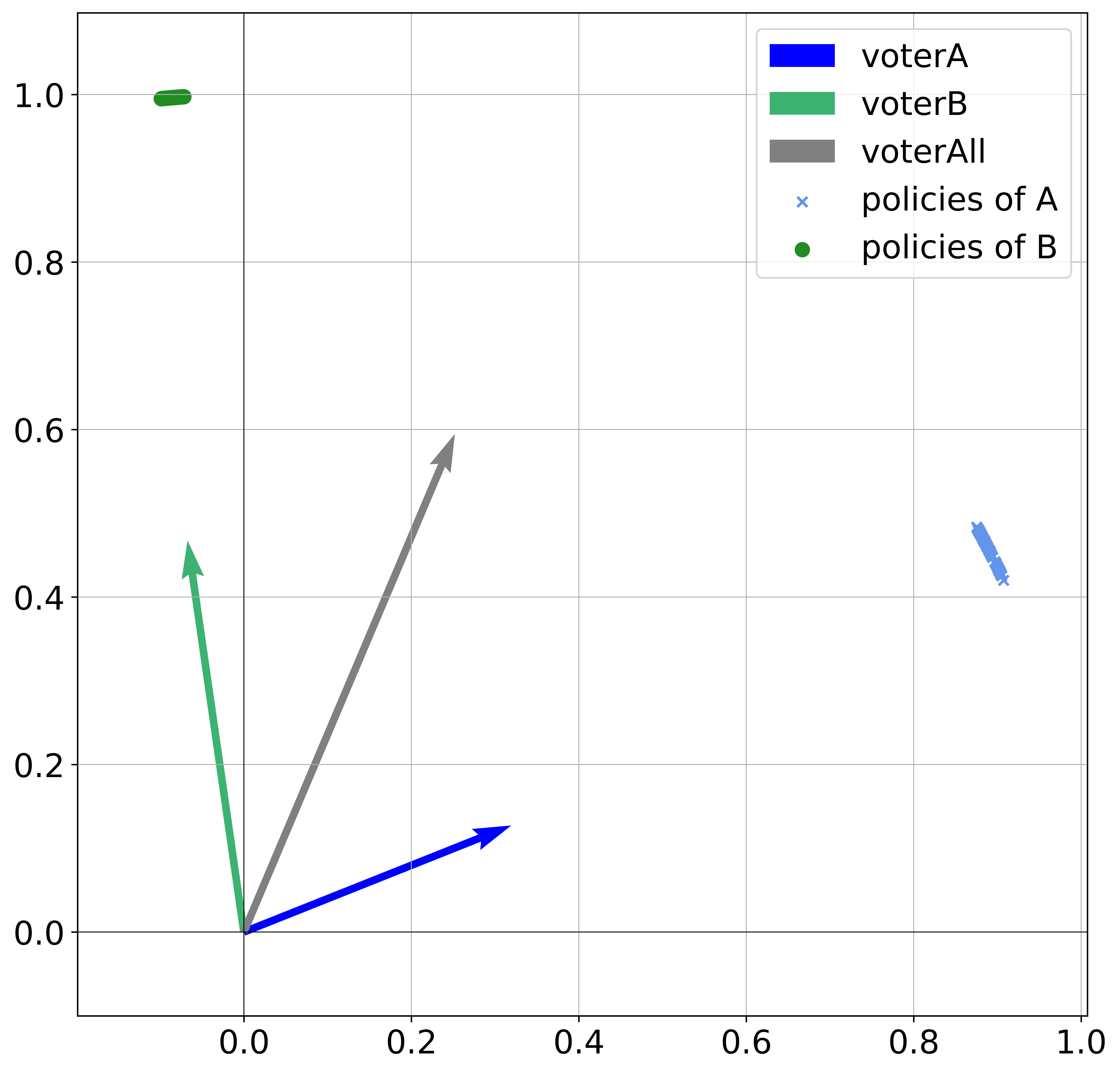}
    \end{subfigure}
    \hfill
    \begin{subfigure}[b]{0.45\textwidth}
    \includegraphics[width=0.85\textwidth]{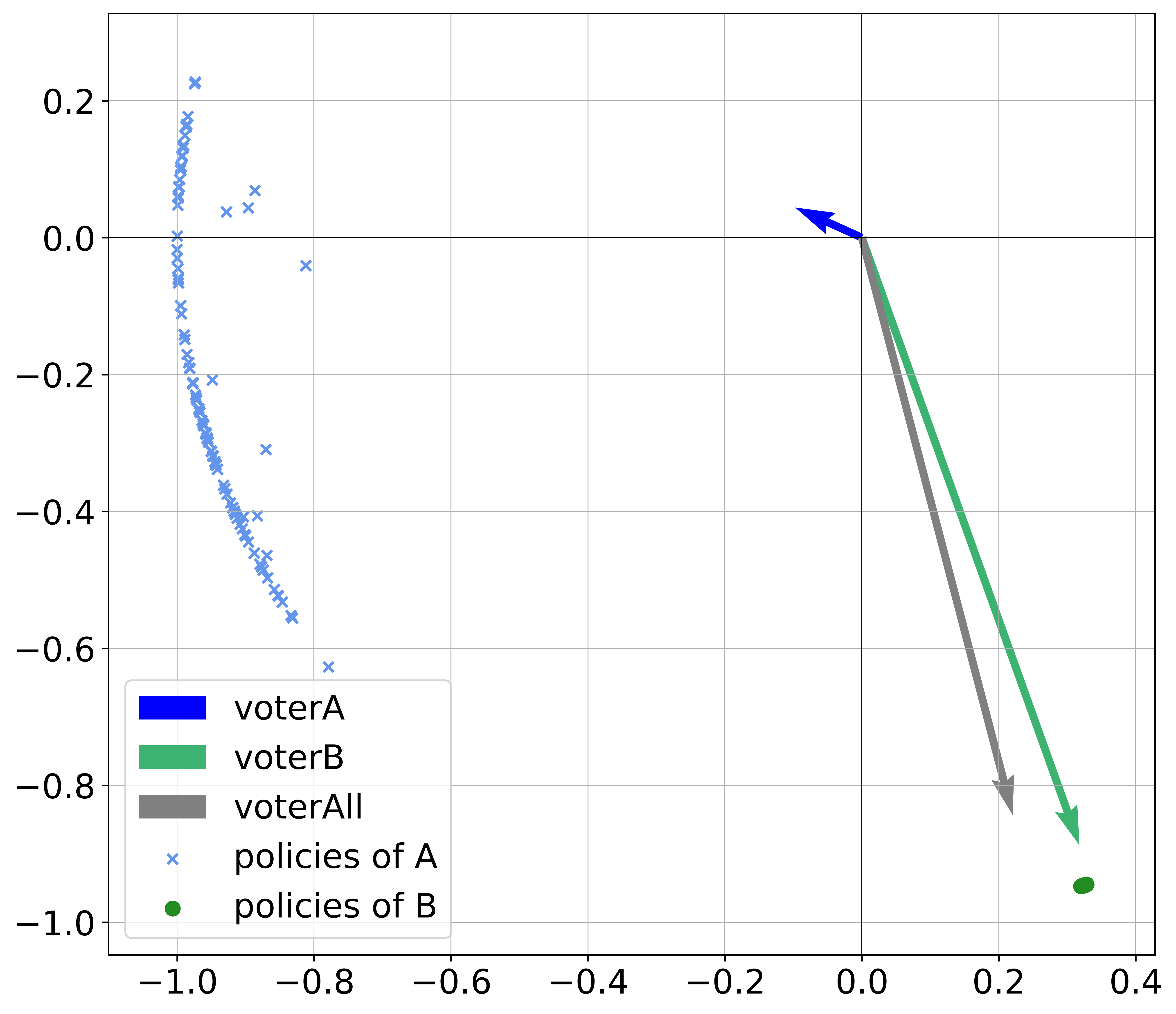}   
    \end{subfigure}
    \caption{Multiple convergence of decentralized projected gradient ascent. Upper figure: under consensus-reachable condition; lower figure: under the non-consensus-reachable condition. Through starting from a random policy profile under both conditions, the decentralized projected gradient descent does not converge to a unique profile. Arrows represents~$\bm{Q_A}$, $\bm{Q_B}$ and~$\bm{Q}$ for the voters' preference vectors of $V_A$, $V_B$ and $V$ respectively. Scatter plots denote the policy profiles in the convergence.}
    \label{fig:GA_experiments_multi-convergence}
\end{figure}

\section{Discretization Comes to the Rescue}
\label{sec:discretization}

Since our strategy domain is bounded, naturally it suggests that discretization might be helpful for getting an approximate PSNE. 
Consider the pseudo-gradient $F(\theta_A, \theta_B) = (F_1(\theta_A, \theta_B), F_2(\theta_A, \theta_B))$ of the two players' payoff functions with respect to $\theta_A$ and $\theta_B$. The feasible space~$[0,\rho_A]\times [0, \rho_B]$ of~$(\theta_A, \theta_B)$ can be discretized using a Lipschitz constant $\mathcal{L}$ of~$F$, 
and search for an approximate PSNE in a grid of size depending on~$\mathcal{L}$. 

Using the chain rule and applying $\odif{x}/\odif{\theta_A} = -\sin(\theta_A) = -\sqrt{1-x^2}$ and $\odif{y}/\odif{\theta_B} = -\sin(\theta_B) = -\sqrt{1-y^2}$, we derive the first derivatives of~$p_A$ with respect to $\theta_A$ and $\theta_B$ as  
$\frac{\partial p_{A}}{\partial \theta_A} = \frac{\|Q\|}{8}\Bigl(-C_{1}\sqrt{1-x^2} +x\sqrt{1-C_{1}^{2}}\Bigr)$ and 
$\frac{\partial p_{A}}{\partial \theta_B} = \frac{\|Q\|}{8}\Bigl(C_{2}\sqrt{1-y^2} -y\sqrt{1-C_{2}^{2}}\Bigr)$, 
respectively. So, the pseudo-gradient $F$ in terms of $\theta_A$ and $\theta_B$ can be written as  
$F_{1}(\theta_A,\theta_B) = 
\|Q_{A}\|\Bigl(-p_{A}\sqrt{1-x^2}
+ \bigl(x - M(y)\bigr)\,\pdv{p_{A}}{\theta_A}\Bigr)$, 
$F_{2}(\theta_A,\theta_B) = 
\|Q_{B}\|\Bigl((1 - p_{A})(-\sqrt{1-y^2})
+ \bigl(M(x) - y\bigr)\,\pdv{p_{A}}{\theta_B}\Bigr)$. 
To find the Lipschitz bound for the payoff function $R_A(\theta_A, \theta_B)$, 
we aim to find $\max_{\theta_A,\theta_B}|F_1(\theta_A, \theta_B)|$. 
Note that by the definition of~$K$ and~$L$, we have 
$M(y) = Ky+L\sqrt{1-y^2} = \cos(\rho_A+\rho_B)\cos(\theta_B)+\sin(\rho_A+\rho_B)\sin(\theta_B) = \cos(\rho_A+\rho_B-\theta_B)\in [-1, 1]$, 
and similarly $M(x)\in [-1, 1]$. Thus, we obtain   
    $|F_1(\theta_A, \theta_B)| \leq \|Q_A\|[|p_A| + |x-M(y)| \left|\pdv{p_A}{\theta_A}\right| \Bigr]\leq 2\|Q_A\|$, 
in which we use the fact that $\left| \pdv[]{p_A}{\theta_A} \right| \leq \frac{\|Q\|}{8} \left(|C_1| + \sqrt{1-C_1^2}\right)\leq \frac{\|Q\|}{4}\leq \frac{1}{2}$. 
Similarly, we have $|F_2(\theta_A, \theta_B)| \leq 2\|Q_B\|$, so we obtain $\mathcal{L}\leq 2(\|Q_A\|+\|Q_B\|)$. 
Thus, we devise a grid-based searching algorithm \texttt{GBA-PSNE} to find an approximate PSNE of the game.

\begin{algorithm}[ht]
\caption{Grid-Based Approx. $\varepsilon$-PSNE Search (\texttt{GBA-PSNE})}
\begin{algorithmic}[1]
\REQUIRE $Q_A,Q_B,Q$, accuracy parameter $\varepsilon>0$
\ENSURE a grid profile that is an $\varepsilon$-PSNE of the continuous game

\STATE Compute the Lipschitz constant $\mathcal{L}$ for $f_{\theta_A}$ and $g_{\theta_B}$.
\STATE Set the grid spacing $h$ so that $h\leq \varepsilon/(4\mathcal{L})$, and let $\hat{\varepsilon}\leftarrow 3\mathcal{L}h$.
\STATE Construct the uniform grids with spacing~$h$
\[
G_A=\{a_0,a_1,\ldots,a_{N-1}\}\subseteq [0,\rho_A],
\]
\[
G_B=\{b_0,b_1,\ldots,b_{N-1}\}\subseteq [0,\rho_B].
\]

\FOR{$j=0$ \TO $N-1$}
    \STATE Compute $M_A(j):=\max_{0\leq i\leq N-1} f_{\theta_A}(a_i,b_j)$;
    \STATE Compute the interval $I_A(j)$ of all indices $i$ satisfying
    $f_{\theta_A}(a_i,b_j)\geq M_A(j)-\hat{\varepsilon}$.
\ENDFOR

\FOR{$i=0$ \TO $N-1$}
    \STATE Compute $M_B(i):=\max_{0\leq j\leq N-1} g_{\theta_B}(a_i,b_j)$;
    \STATE Compute the interval $I_B(i)$ of all indices $j$ satisfying
    $g_{\theta_B}(a_i,b_j)\geq M_B(i)-\hat{\varepsilon}$.
\ENDFOR

\STATE Run Algorithm~\texttt{Sweep-Intervals} on the family of intervals
$\{I_A(j)\}_{j=0}^{N-1}$ and $\{I_B(i)\}_{i=0}^{N-1}$.
\STATE Let $(i,j)$ be the returned pair and return $(a_i,b_j)$
\end{algorithmic}
\end{algorithm}

\begin{algorithm}[t]
\caption{\!\!\texttt{Sweep-Intervals}$(\{I_A(j)\}_{j=0}^{N-1}, \,\{I_B(i)\}_{i=0}^{N-1})$}
\label{alg:Sweep-Intervals}
\begin{algorithmic}[1]
\REQUIRE $I_A(j)=[I^{|\leftarrow}_A(j),I^{\rightarrow|}_A(j)]$ \, $\forall j$;\, 
$I_B(i)=[I^{|\leftarrow}_B(i),I^{\rightarrow|}_B(i)]$ \, $\forall i$
\ENSURE a pair $(i,j)$ such that $i\in I_A(j)$ and $j\in I_B(i)$, if any.
\STATE Initialize empty lists $\mathrm{Add}[0],\ldots,\mathrm{Add}[N-1]$
and $\mathrm{Rem}[0],\ldots,\mathrm{Rem}[N-1]$.

\FOR{$i=0$ \TO $N-1$}
    \STATE Append $i$ to $\text{Add}[I^{|\leftarrow}_B(i)]$.
    \IF{$I^{\rightarrow|}_B(i)+1\leq N-1$}
        \STATE Append $i$ to $\text{Rem}[I^{\rightarrow|}_B(i)+1]$.
    \ENDIF
\ENDFOR

\STATE Initialize an empty ordered set $\mathcal{A}$.

\FOR{$j=0$ \TO $N-1$}
    \FOR{$i\in \text{Rem}[j]$}
        \STATE delete $i$ from $\mathcal{A}$.
    \ENDFOR
    \FOR{$i\in \text{Add}[j]$}
        \STATE insert $i$ into $\mathcal{A}$.
    \ENDFOR
    \STATE \COMMENT{$\mathcal{A} = \{ i\mid j\in I_B(i)\}$: the set of active indices at step~$j$}
    \STATE Let $t\gets \min\{ i\in\mathcal{A}\mid i\geq I^{|\leftarrow}_A(j)\}$.
    
    \IF{$t\neq \text{null}$ and $t\leq I^{\rightarrow|}_A(j)$}
    \STATE \textbf{return} $(t,j)$
    \ENDIF
\ENDFOR

\STATE \textbf{return} $\emptyset$ \COMMENT{``no pair found''}
\end{algorithmic}
\end{algorithm}

\begin{algorithm}[ht]
\caption{Ternary Best Response (\texttt{TBR})}
\begin{algorithmic}[1]
\REQUIRE Grid $\{z_k\}_{k=0}^{N-1}$, 
index interval $[L,R] \subseteq \{0,\dots,N-1\}$,
unimodal objective $f$. 
\ENSURE $k^{\ast} = \argmax_{k\in[L,R]} f(k)$
\WHILE{$R-L>2$}
  \STATE $m_1 \gets L + \big\lfloor \frac{R-L}{3} \big\rfloor,\;\;
         m_2 \gets R - \big\lfloor \frac{R-L}{3} \big\rfloor$
  \IF{$f(m_1) \leq f(m_2)$}
     \STATE $L \gets m_1$
  \ELSE
     \STATE $R \gets m_2$
  \ENDIF
\ENDWHILE
\STATE return $\arg\max_{k\in\{L,L+1,\ldots,R\}} f(k)$
\end{algorithmic}
\end{algorithm}

\begin{lem}\label{lem:TBR}
Algorithm \texttt{TBR} can find the maximizer of an unimodal function $f$ in $O\!\big(\log\big(\tfrac{b-a}{\varepsilon}\big)\big)$ steps. 
\end{lem}
\begin{proof}
Each iteration in the while-loop shrinks the interval length by a constant factor~$2/3$, since
for unimodal $f$ one of the two outer thirds cannot contain the maximizer.
After $t=O(\log((b-a)/\varepsilon))$ iterations, we evaluate at most three remaining 
indices to return the exact grid maximizer.
\end{proof}

\begin{thm}\label{thm:grid-approx}
Algorithm \texttt{GBA-PSNE} finds an $\varepsilon$-PSNE in~$O(nk+ (k/\epsilon)\log(1/\varepsilon))$ time. 
\end{thm}

\begin{proof}
(sketch) 
Given the grid $G_A\times G_B$ from the algorithm, 
with a slight abuse of notation, let $(a_i, b_j) \in G_A \times G_B$ be a strategy profile that is an $\hat{\varepsilon}$-PSNE \emph{on the grid} such that
$R_A(a_i, b_j):=f_{\theta_A}(a_i, b_j) \geq \max_{a_{\ell} \in G} f_{\theta_A}(a_{\ell}, b_j) - \hat{\varepsilon} \mbox{ and }
R_B(a_i, b_j) := g_{\theta_B}(a_i, b_j)\allowbreak \geq \max_{b_\ell \in G} g_{\theta_B}(a_i, b_\ell) - \hat{\varepsilon}$.
Since $h \leq \frac{\varepsilon}{2\mathcal{L}}$ and $\hat{\varepsilon} = \varepsilon - \mathcal{L}h$, we can derive that $(a_i, b_j)$ is an $\varepsilon$-PSNE in the game.
The time complexity of the algorithm can be shown to be~$O(nk+ k/\varepsilon^2)$. In addition, since $R_A$ and $R_B$ are one-dimensional quasi-concave functions for fixed $r_A=r_B=1$, they are unimodal so that Lines~5,6,9,10 of Algorithm \texttt{GBA-PSNE} can be computed using Algorithm~\texttt{TBR} for the ternary search and 
binary searches for the endpoints of intervals of best-responses. The original grid-search time $O(N^2) = O(1/\epsilon^2)$ can then be shrunk to $O(N\log N) = O((1/\epsilon)\log(1/\epsilon))$. 

\end{proof}

\vspace{-17pt}
\section{Concluding Remarks} 
\label{sec:future}

In this paper, we 
extend prior work by conceptualizing real-valued policies within a compact one or multi-dimensional Euclidean space, transforming the problem from discrete choices to a continuous optimization. The isotonicity hypothesis concerning winning probability is validated through voting simulations via different functions voters refer to cast their votes, consistently demonstrating a monotonic relationship between utility difference and winning probability. Furthermore, we establish the existence of PSNE under various conditions: it provides a closed-form solution for the one-dimensional policy space and formally proves PSNE existence in the multi-dimensional setting. 
Our work also serves as an extension beyond the classic result that multi-dimensional Downsian competition which needs not admit any equilibrium~\cite{plott1967,xefteris2017}, and complements previous distance-based spatial models by considering the direction and strength of policies with respect to the voters. Our model predicts max-intensity platforms with directional moderation which captures one facet of extreme politics.

Beyond existence, this work offers algorithmic insights. We demonstrate that the game's pseudo-gradient mapping is generally not monotone, disproving cocoercivity and suggesting that standard convergence guarantees for gradient-based methods do not apply. Despite this theoretical challenge, experimental simulations show that a decentralized projected gradient ascent algorithm converges rapidly to an approximate PSNE in most cases. 
To provide a guaranteed approach, a polynomial-time grid-based search algorithm (\texttt{GBA-PSNE}) is proposed, which systematically finds an $\varepsilon$-approximate PSNE by discretizing the policy domain and leveraging a derived Lipschitz constant. These findings collectively contribute to the analysis of continuous non-cooperative games and offer practical frameworks for understanding stable policy outcomes in political competition.

We suggest promising directions for future research below. 
\paragraph{Winning probability functions.} 
Inspired by the dueling bandit setting, we choose a linear function for computing the winning probabilities. It will be interesting to consider other isotonic winning probability functions and extend our findings herein. 
\paragraph{Multi-party competitions.} 
Political competitions can involve three or more competitors under uncertainty, so it would be interesting to explore how our framework can be extended to such scenarios. 
\paragraph{Beyond full information.} 
When voters' preferences are unobservable, policy making requires the elicitation of voters’ preferences, and voters may vote strategically. It would be interesting to see how the isotonicity assumption and the following analysis extend to such settings.



\begin{acks}
This work is funded and supported by the National Science and Technology Council, Taiwan, under grant nos. NSTC 112-2221-E-032-018-MY3, NSTC~113-2628-HA49-002-MY3, and NSTC~114-2221-E005-058. This work is also partially supported by NSTC 114-2634-F-005-002- through the Smart Sustainable New Agriculture Research Center (SMARTer). 
\end{acks}



\bibliographystyle{ACM-Reference-Format} 
\bibliography{policy_competition}

\end{document}